\documentclass{llncs}

\usepackage{caption,subcaption,graphicx,paralist,enumerate}
\usepackage{xspace}
\usepackage[textsize=footnotesize]{todonotes}
\usepackage[basic]{complexity}

\usepackage{amsmath,amssymb,amsthm}
\spnewtheorem{observation}{Observation}{\bfseries}{\itshape}

\usepackage[pdfpagelabels,colorlinks,citecolor=blue,linkcolor=black,urlcolor=blue]{hyperref}
\usepackage[nameinlink,capitalize]{cleveref}
\Crefformat{figure}{#2Fig.~#1#3}
\crefname{property}{Property}{Properties}
\crefname{property2}{Property}{Properties}
\crefname{theorem2}{Theorem}{Theorems}
\crefname{lemma2}{Lemma}{Lemmas}
\crefname{observation}{Observation}{Observations}
\crefname{observation2}{Observation}{Observations}

\usepackage[bibliography=common]{apxproof}

\newtheoremrep{theorem2}[theorem]{Theorem}
\newtheoremrep{corollary2}[corollary]{Corollary}
\newtheoremrep{lemma2}[lemma]{Lemma}
\newtheoremrep{property2}[property]{Property}
\newtheoremrep{observation2}[observation]{Observation}

\renewcommand{\orcidID}[1]{\href{https://orcid.org/#1}{\includegraphics[scale=.03]{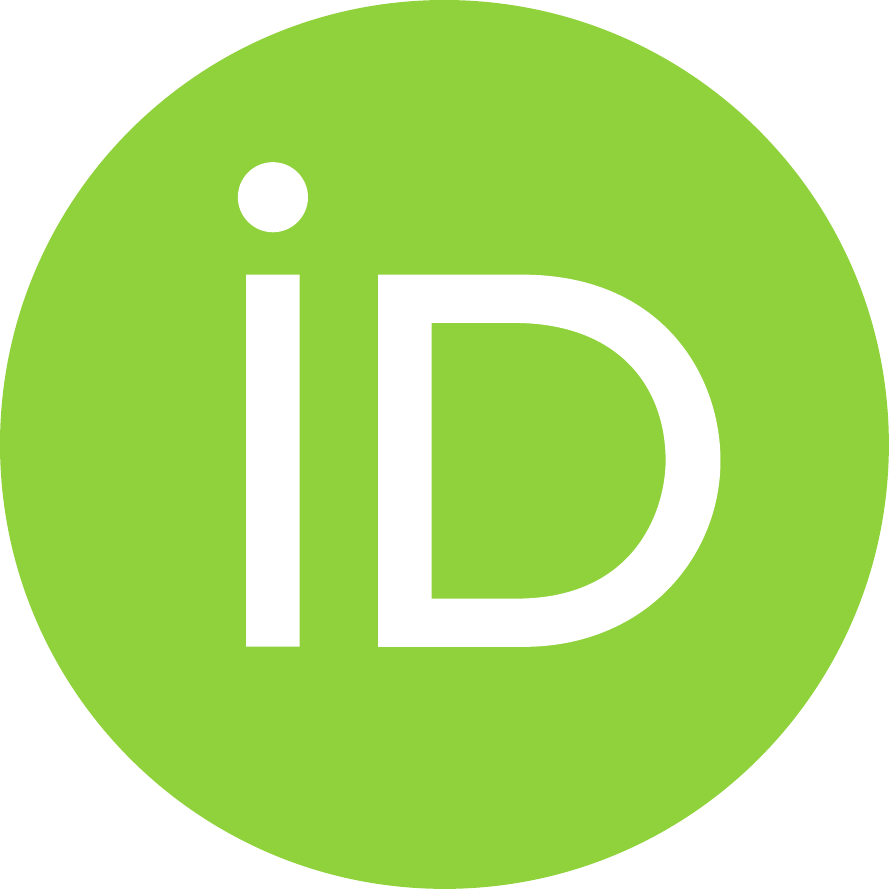}}}
\renewcommand{\emph}[1]{{\em \textcolor{blue}{#1}}\xspace}
\newcommand{\myparagraph}[1]{\medskip\noindent\textbf{#1}}

\newcommand{\fabristar}{$\star$}

\title{Strictly-Convex Drawings\\of $3$-Connected Planar Graphs\thanks{Research of FM partially supported by Dip. Ingegneria, Univ. of Perugia, grant RICBA21LG ``Algoritmi, modelli e sistemi per la rappresentazione visuale di reti''.}}

\titlerunning{Strictly-Convex Drawings of $3$-Connected Planar Graphs}

\author{Michael~A.~Bekos\inst{1}\orcidID{0000-0002-3414-7444} \and
Martin~Gronemann\inst{2}\orcidID{0000-0003-2565-090X} \and \\
Fabrizio Montecchiani\inst{3}\orcidID{0000-0002-0543-8912} \and Antonios~Symvonis\inst{4}\orcidID{0000-0002-0280-741X}}

\authorrunning{M.~A.~Bekos, M.~Gronemann, F.~Montecchiani, A.~Symvonis}

 \institute{Department of Mathematics, University of Ioannina, Ioannina, Greece\\
 \email{bekos@uoi.gr}
 \and  Algorithms and Complexity Group, TU Wien, Vienna, Austria\\
 \email{mgronemann@ac.tuwien.ac.at}
 \and Department of Engineering, University of Perugia, Italy\\
 \email{fabrizio.montecchiani@unipg.it}
 \and School of Applied Mathematical \& Physical Sciences, NTUA, Greece\\
 \email{symvonis@math.ntua.gr}}

\begin{document}
\maketitle 

\begin{abstract}
Strictly-convex straight-line drawings of $3$-connected planar graphs in small area form a classical research topic in Graph Drawing. Currently, the best-known area bound for such drawings is $O(n^2) \times O(n^2)$, as shown by B\'{a}r\'{a}ny and Rote by means of a sophisticated technique based on perturbing (non-strictly) convex drawings. Unfortunately, the hidden constants in such area bound are in the $10^4$ order. 

We present a new and easy-to-implement technique that yields strictly-convex straight-line planar drawings of $3$-connected planar graphs on an integer grid of size $2(n-1) \times (5n^3-4n^2)$.
\end{abstract}

%======================================================
\section{Introduction}
%======================================================

Drawing planar graphs is a fundamental topic in Graph Drawing with several important contributions over the last few decades~\cite{DBLP:journals/combinatorica/FraysseixPP90,Far48,DBLP:journals/algorithmica/Kant96,DBLP:conf/soda/Schnyder90,Tu63}. One of the most influential is due to de~Fraysseix, Pach and Pollack~\cite{DBLP:journals/combinatorica/FraysseixPP90}, who back in 1988 showed that every $n$-vertex planar graph admits a straight-line planar drawing on a $(2n-4)\times(n-2)$ grid, which can be computed in $O(n)$ time~\cite{DBLP:journals/ipl/ChrobakP95}. Since then, several improvements on the size of the underlying grid have been proposed in the literature~\cite{DBLP:journals/order/Felsner01,DBLP:journals/dcg/He97,DBLP:journals/algorithmica/Kant96,DBLP:journals/dcg/MiuraNN01,DBLP:conf/soda/Schnyder90,DBLP:conf/wads/ZhangH03}. The best-known upper bound is $(n-2)\times(n-2)$ by Chrobak and Kant~\cite{DBLP:journals/ijcga/ChrobakK97} and by Schnyder~\cite{DBLP:conf/soda/Schnyder90}, who propose two conceptually different approaches to derive this bound. The former is an incremental drawing algorithm inspired by~\cite{DBLP:journals/combinatorica/FraysseixPP90}, while the latter is based on a face counting technique.

Straight-line drawings of planar graphs have also been extensively studied by requiring convexity~\cite{DBLP:journals/SchnyderT92}, that is, the boundary of every face must be a convex polygon. Such drawings are called \emph{convex} and always exist for $3$-connected planar graphs~\cite{DBLP:journals/jct/Thomassen84,Tu63}. Again the aim is to keep the size of the underlying grid as small as possible; see~\cite{Battista2013} for a survey. Early results date back to Schnyder and Trotter~\cite{DBLP:journals/SchnyderT92}, Chrobak and Kant~\cite{DBLP:journals/ijcga/ChrobakK97}, Di Battista et al.~\cite{DBLP:journals/algorithmica/BattistaTV99} and Felsner~\cite{DBLP:journals/order/Felsner01}. The latter guarantees the existence of a convex drawing of a $3$-connected planar graph~on a $(f-1)\times(f-1)$ integer grid, where $f=O(n)$ is the number of faces.

Note that in a convex drawing three vertices on the boundary of a face can be collinear. If this is not allowed, then the corresponding drawings are called \emph{strictly-convex}. Since an $n$-vertex cycle cannot be drawn strictly-convex on a grid of size $o(n^3)$~\cite{DBLP:conf/compgeom/ChrobakGT96}, it follows that strictly-convex drawings are more demanding in terms of required area. 
As an adaptation of the standard incremental drawing algorithms or the face-counting methods is rather difficult, the only approach that has been exploited so far to obtain strictly-convex drawings is to perturb convex drawings. This idea was pioneered by Chrobak, Goodrich and Tamassia~\cite{DBLP:conf/compgeom/ChrobakGT96}, who claimed (without giving details) that every $3$-connected planar graph admits a strictly-convex drawing on an $O(n^3) \times O(n^3)$ grid. The area bound was improved
to $O(n^{7/3}) \times O(n^{7/3})$ by Rote~\cite{DBLP:conf/soda/Rote05} and to $O(n^2) \times O(n^2)$ by B\'{a}r\'{a}ny and Rote~\cite{Barany2006}, which is currently the best-known asymptotic upper bound.
However, as the authors mention ``the constants hidden in the $O$-notation are on the order of 100 for the width and on the order of 10,000 for the height. This is far too much for applications where one wants to draw graphs on a computer screen''~\cite{Barany2006}.

\begin{figure}[t]
    \centering
    \includegraphics{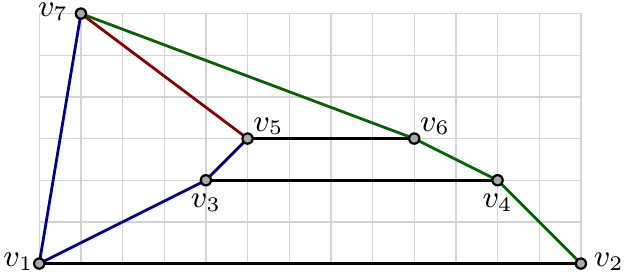}
    \caption{A strictly-convex drawing of a 3-connected planar graph on $7$ vertices.}
    \label{fig:sample}
\end{figure}

\myparagraph{Our contribution.} We continue the research on strictly-convex drawings of $3$-connected planar graphs. 
Our contribution is a new technique that computes strictly-convex drawings of $3$-connected $n$-vertex planar graphs on an integer grid of size $2(n-1) \times (5n^3-4n^2)$, as outlined in the following theorem. Although the asymptotic area bound is the same as the one in~\cite{Barany2006}, the multiplicative constants are significantly smaller. Also, the proposed technique is elegant and can be readily implemented to run in linear time. On the other hand, the aspect-ratio of the produced drawings is quadratic rather than constant.

\begin{theorem}\label{thm:main}
Every $3$-connected planar graph with $n$ vertices admits a strictly-convex planar straight-line  drawing on an integer grid of size $2(n-1) \times (5n^3-4n^2)$. Also, the drawing can be computed in $O(n)$ time.
\end{theorem}
\noindent Our technique starts with a convex drawing computed by Kant's algorithm~\cite{DBLP:journals/algorithmica/Kant96}. We rely on properties of such a drawing to show that shifting vertices upwards  by using a  strictly-increasing and strictly-convex function preserves planarity; a property of independent interest. Also, the obtained planar drawing is convex and collinear vertices in a face, if any, are horizontally aligned. For such vertices, a second shifting yields an internally strictly-convex drawing. A suitable augmentation guarantees that the outer face is also strictly-convex. 

\myparagraph{Paper structure.} \Cref{sec:preliminaries} contains basic definitions and tools. In \Cref{sec:kant}, we introduce properties of Kant's algorithm that we leverage in our technique. \Cref{sec:main} describes our algorithm. \Cref{sec:conclusions} concludes the paper with a brief discussion and open problems. For space reasons, some proofs are in the appendix and the corresponding statements are marked with ($\star$).

%======================================================
\section{Preliminaries}\label{sec:preliminaries}
%======================================================
\begin{figure}[t]
    \centering
    \includegraphics[page=1]{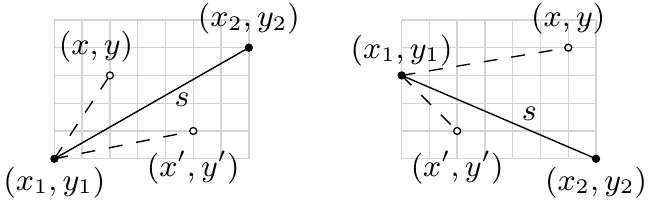}
	\caption{In both drawings, $(x,y)$ is above $s$, while $(x',y')$ is below $s$. }
	\label{fig:above-below}
\end{figure}

\textbf{Basic definitions.} Let $f: \mathbb{R} \rightarrow \mathbb{R}$ be a function. If $f(a) < f(b)$ for every pair $a,b \in \mathbb{R}$ with $a < b$, then $f$ is \emph{strictly-increasing}. Function $f$ is \emph{strictly-convex} if for all $t$, $0< t<1$, and all $a,b \in \mathbb{R}$, it holds $f(ta+(1-t)b) < t f(a)+(1-t)f(b)$. Consider three points $(x_1,y_1), (x,y), (x_2,y_2)$, with $x_1 < x < x_2$ and let $s$ be the line-segment connecting $(x_1,y_1)$ and $(x_2,y_2)$. We say that $(x,y)$ is \emph{above} (resp., \emph{below}) $s$ if the slope of $s$ is smaller (resp., larger) than the slope of the line-segment connecting $(x_1,y_1)$ and $(x,y)$; see \Cref{fig:above-below}. The next lemma easily follows from the chordal slope lemma~\cite{martinbook} (see also \cref{lem:chordal-slope} in the appendix).

\begin{toappendix}
This section is devoted to proving \cref{le:basic}. To this end, the following lemma, known as \emph{chordal slope lemma}, is needed. In the statement, the notation $\text{slope}(p,q)$ denotes the slope of the line-segment with endpoints~$p$~and~$q$.

\begin{lemma}[Chordal slope lemma~\cite{martinbook}]\label{lem:chordal-slope}
Let $f: (a,b) \rightarrow \mathbb{R}$ be a convex function. Consider  $x_1 < x < x_2$ in $(a,b)$, and the points $p_1=(x_1,f(x_1))$, $p=(x,f(x))$, and $p_2=(x_2,f(x_2))$. Then: 
\[
\text{slope}(p_1,p) \le \text{slope}(p_1,p_2) \le \text{slope}(p,p_2).
\]
\end{lemma}

\noindent In our proof below, we use the version of this lemma in which the function $f$ is strictly-convex, which implies that all inequalities of the statement are strict.
\end{toappendix}

\begin{lemma2rep}[$\star$]\label{le:basic}
Let $(x_1,y_1)$, $(x,y)$ and $(x_2,y_2)$ be three collinear points with $x_1 < x < x_2$ that are not horizontally aligned. If $f:\mathbb{R} \rightarrow \mathbb{R}$ is a strictly-convex function, then point $(x, f(y))$ is below the line-segment with endpoints $(x_1,f(y_1))$ and $(x_2,f(y_2))$.
\end{lemma2rep}
\begin{proof}
Since $(x_1,y_1), (x,y), (x_2,y_2)$ are collinear, by Thales' Theorem, we know $\frac{y_2-y_1}{y-y_1}=\frac{x_2-x_1}{x-x_1}$. Suppose $y_1 < y < y_2$, as otherwise, if $y_1 > y > y_2$, the argument is symmetric. Then, the claim holds if $\frac{f(y_2)-f(y_1)}{x_2-x_1}>\frac{f(y)-f(y_1)}{x-x_1}$. We can rewrite this inequality as $\frac{f(y_2)-f(y_1)}{y_2-y_1}>\frac{f(y)-f(y_1)}{y-y_1}$, which in turn follows by the chordal slope lemma given above.
\end{proof}

\myparagraph{Drawings and embeddings.} We assume familiarity with basic graph drawing concepts~\cite{DBLP:books/ph/BattistaETT99}. In particular, a \emph{plane graph} is a graph with a prescribed planar embedding. Unless otherwise specified, we consider drawings that are straight-line, planar and whose vertices are on an integer grid. A drawing is \emph{convex} (\emph{strictly-convex}) if the boundary of each face is a convex (strictly-convex) polygon. Similarly, a drawing is \emph{internally convex} (\emph{internally strictly-convex}) if the boundary of each inner face is a convex (strictly-convex) polygon. 
Given a drawing $\Gamma$ of a graph $G$, denote by $(x_u,y_u)$ the coordinates of vertex $u$ in $\Gamma$. For two vertices $u$ and $v$ in $\Gamma$, we denote by $\Delta_{uv}$ the interior of the right triangle whose corners are $u$, $v$ and the intersection of the vertical line though the vertex having the lowest $y$-coordinate with the horizontal line through the vertex having the highest $y$-coordinate (among $u,v$). For example, in \cref{fig:contour-condition}, the $\Delta_{uv}$ triangle of the endpoints of each edge $(u,v)$~is~striped.

\myparagraph{Canonical order.}  Let $G$ be a $3$-connected plane graph with $n$ vertices. Let $\delta = (P_0,\ldots,P_m)$ be a partition of the vertices of $G$ into paths, such that $P_0 = \{v_1,v_2\}$, $P_m=\{v_n\}$, and edges $(v_1,v_2)$ and $(v_1,v_n)$ exist and belong to the outer face. For $k=0,\ldots,m$, let $G_k$ be the subgraph induced by $\cup_{i=0}^k P_i$. Let $C_k$ be the \emph{contour} of $G_k$ defined as follows: If $k=0$, then $C_0$ is the edge $(v_1,v_2)$, while if $k>0$, then $C_k$ is the path from $v_1$ to $v_2$ obtained by removing $(v_1,v_2)$ from the cycle delimiting the outer face of $G_k$. Partition $\delta$ is a \emph{canonical order}~\cite{DBLP:journals/algorithmica/Kant96}  of $G$ if for each $k=1,\ldots,m-1$ the following conditions hold: %
\begin{inparaenum}[(i)]
\item\label{co:1}$G_k$ is biconnected and internally $3$-connected,
\item\label{co:2}all neighbors of $P_k$ in $G_{k-1}$ are on $C_{k-1}$,
\item\label{co:3}either $P_k$ is a \emph{singleton} (i.e., $|P_k|=1$), or $P_k$ is a \emph{chain} (i.e., $|P_k|>1$) and the degree of each vertex of $P_k$ is $2$ in $G_k$,
\item\label{co:4}all vertices of $P_k$ with $0\leq k < m$ have at least one neighbor in $P_j$ for some $j > k$. For example, a canonical order for the graph of \cref{fig:sample} is $P_0=\{v_1,v_2\}$, $P_1=\{v_3,v_4\}$, $P_2=\{v_5,v_6\}$ and $P_3=\{v_7\}$.
\end{inparaenum}
A canonical order of $G$ can be computed in $O(n)$ time~\cite{DBLP:journals/algorithmica/Kant96}. 

\myparagraph{Kant's algorithm.} Kant~\cite{DBLP:journals/algorithmica/Kant96} describes an incremental drawing algorithm that, in linear time,  computes a convex straight-line planar drawing $\Gamma$ of an $n$-vertex plane graph $G$ on an integer grid of size $(2n-4) \times (n-2)$. The drawing $\Gamma$ has the same planar embedding as the input graph $G$. The algorithm is based on a canonical order $\delta$ of $G$ and works as follows: Initially, vertices $v_1$ and $v_2$ of $P_0$ are placed at points $(0,0)$ and $(1,0)$, respectively. For $k=1,\ldots,m$, assume that a  convex  drawing $\Gamma_{k-1}$ of $G_{k-1}$ has been constructed in which the edges of contour $C_{k-1}$ are drawn with slopes $0$ and $\pm 1$ (\emph{contour condition}; see \cref{fig:contour-condition}). 
Let $(w_1,\ldots,w_p)$ be the vertices of  $C_{k-1}$ from left to right in $\Gamma_{k-1}$, where $w_1=v_1$ and $w_p=v_2$. 
Each vertex $v$ in $G_{k-1}$ has been associated with a \emph{shift-set} $S(v)$, such that $\Gamma_{k-1}$ is \emph{stretchable}, that is, for each $i=1,\ldots,p$ the result of shifting $S(w_i),\ldots,S(w_p)$ by one (or more) units to the right is a convex drawing of $G_{k-1}$. 
Let $P_k = \{z_1,\ldots,z_p \}$ be the next path in $\delta$. Let $w_\ell$ and $w_r$ be the leftmost and rightmost neighbors of $P_k$ on $C_{k-1}$ in $\Gamma_{k-1}$, where $1 \leq  \ell < r \leq p$. 
To introduce $P_k$ and to avoid edge-overlaps, the algorithm first identifies two so-called critical vertices $w_{\ell'}$ and $w_{r'}$ with $\ell \leq \ell',r' \leq r$ and then shifts~(i)~by one unit to the right each vertex in $\bigcup_{i=\ell'}^{p} S(w_i)$ and then
(ii)~by one unit to the right each vertex in $\bigcup_{i=r'}^{p} S(w_i)$. Then, $z_1$ is placed at intersection of the line of slope $+1$ through $w_\ell$ with the line through of slope $-1$ point $w_r$; see \Cref{fig:placement}. If $P_k$ is~a chain, then for $i=2,\ldots,p$, vertex $z_i$ is placed one unit to the right of $z_{i-1}$ by shifting  each vertex in $\bigcup_{i=r'}^{p} S(w_i)$ one unit to the right. Finally, the shift-sets of the vertices of $P_k$ are defined accordingly to ensure that $\Gamma_k$ is stretchable. 

\begin{figure}[t!]
	\centering
	\begin{subfigure}{.48\textwidth}
	\flushleft
	\includegraphics[page=1,scale=0.85]{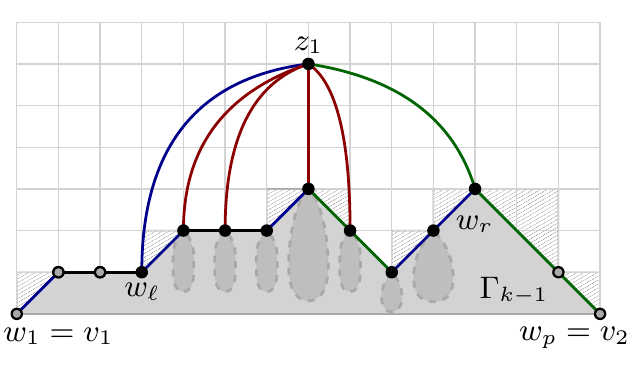}
	\subcaption{Contour condition}
	\label{fig:contour-condition}
	\end{subfigure}
	\hfil
	\begin{subfigure}{.48\textwidth}
	\flushright
	\includegraphics[page=2,scale=0.85]{figs/kant.pdf}
	\subcaption{Placement of $z_1$ in $\Gamma_{k-1}$}
	\label{fig:placement}
	\end{subfigure}
	\caption{%
	Introducing a singleton $P_k=\{z_1\}$ in $\Gamma_{k-1}$ in the algorithm by Kant~\cite{DBLP:journals/algorithmica/Kant96}.}
	\label{fig:kant}
\end{figure}

%======================================================
\section{Properties of Kant's algorithm}\label{sec:kant}
%======================================================

We provide properties of drawings computed by Kant's algorithm that we leverage in the next section; some of these properties are indirectly mentioned also in~\cite{DBLP:conf/wg/Kant92}. To ease the presentation, we first introduce a $4$-coloring for the edges of $G$ similar to the one by Schnyder~\cite{felsner,DBLP:conf/soda/Schnyder90}. We color edge $(v_1,v_2)$ of $G_0$ black. Given a $4$-coloring for $G_{k-1}$ with $k=1,\ldots,m$, we extend it for $G_k$ as follows (see \cref{fig:sample,fig:contour-condition}). We first color the edges of $G_{k}$ that do not belong to $G_{k-1}$ and are on contour $C_{k}$. Namely, the first such edge encountered in a clockwise walk of $C_k$ from $v_1$ to $v_2$ is blue, the last one is green and all remaining ones (that is, those having both endpoints in $P_k$ when $P_k$ is a chain) are black. The remaining edges of $G_k$ not in $G_{k-1}$ are red (i.e., those that are incident to $P_k$ and are not part of $C_k$; this case only arises if $P_k$ is a singleton by Condition~(\ref{co:3}) of the canonical order), which implies that $C_k$ has no red edges.

Since a shift to introduce a path of $\delta$ in the incremental construction of $\Gamma$ can only decrease the slope of a blue edge, increase the slope of a green edge, while the black and the red edges maintain their slope~\cite{DBLP:journals/algorithmica/Kant96}, we have that:
\begin{itemize}[--]
\item the slope of each blue edge ranges in $(0,1]$,
\item the slope of each black edge is $0$, 
\item the slope of each green edge ranges in $[-1,0)$, and
\item the slope of each red edge ranges in the complement of $[-1,1]$.
\end{itemize}

\noindent Since each inner face in $\Gamma$ is formed when a path of $\delta$ is introduced during the incremental construction, part of it belongs to the contour, while its remaining part is formed by the introduced path, which gives rise to the following property.

\begin{property2}\label{prp:face}
Let $x$ be the leftmost vertex of an inner face $g$ in $\Gamma$ (in case of more than one such vertices select the bottommost one). A counterclockwise walk of $g$ starting from $x$ consists of the following boundary parts (see \cref{fig:face}): 
\begin{enumerate}[i.]
\item\label{f:1}a (possibly empty) strictly descendant path of green edges, 
\item\label{f:2}at most one black edge,
\item\label{f:3}a (possibly empty) strictly ascendant path of blue edges,
\item\label{f:4}a green or red edge,
\item\label{f:5}a (possibly empty) horizontal path of black edges, and 
\item\label{f:6}a blue or red edge.
\end{enumerate}
\end{property2}

\begin{figure}[t]
    \centering
    \includegraphics[page=1]{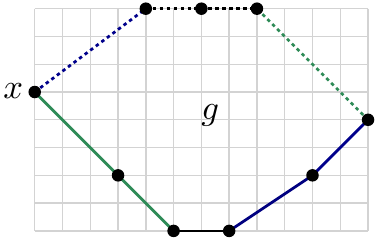}
	\caption{Illustration for the shape of a face.}
	\label{fig:face}
\end{figure}

\noindent Boundary parts~(\ref{f:4})--(\ref{f:6}) (dotted in \cref{fig:face}) are introduced when a path is added during the incremental construction of $\Gamma$ (which implies that at least one of boundary parts~(\ref{f:1})--(\ref{f:3}) is part of the contour and thus is non empty in $g$). So, boundary parts~(\ref{f:4})--(\ref{f:6}) cannot simultaneously contain black and red edges  (by the edge-coloring and Condition~(\ref{co:3}) of canonical order). 

\begin{property2rep}[\fabristar]\hspace{-0.5em}\label{prp:b} 
Each vertex $w$ of a path $P_i$, with $0< i \leq m$, has at~least~two incident edges $(a,w)$ and $(b,w)$, such that $y_a \le y_w$,  $y_b \le y_w$, and $x_a < x_w < x_b$.
\end{property2rep}
\begin{proof}
If $P_i$ is a singleton, then $a$ and $b$ are two extreme neighbors of $w$ in $C_{i-1}$. Otherwise, $P_i$ is a chain $(z_1,\ldots,z_p)$ attached to two vertices $w_\ell$ and $w_r$ along $C_{i-1}$. If $w$ is $z_1$ ($z_p$), then vertex $a$ is $w_\ell$ ($z_{p-1}$), while vertex $b$ is $z_2$ ($w_r$). Otherwise, $w$ is an internal vertex $z_j$ of $P_i$, in which case vertex $a$ is $z_{i-1}$, while vertex $b$ is $z_{i+1}$. By the contour condition, the claimed property holds when $P_i$ is introduced in $\Gamma_i$. Since shiftings performed to derive subsequent drawings do not affect the $y$-coordinate of the vertices in $P_i$ and do not decrease the horizontal distance of already placed vertices, the proof of the property follows. 
\end{proof}

\begin{property2rep}[\fabristar]\hspace{-0.5em}\label{prp:vertical-edges} 
Every face in $\Gamma$ has at most one edge drawn vertical.
\end{property2rep}
\begin{proof}
A shifting cannot introduce a vertically drawn edge in $\Gamma$. Hence, an edge $e$ drawn as a vertical segment is introduced in $\Gamma$ due to a path $P_i$ of $\delta$. Since edges drawn on the contour $C_i$ are of slopes $0$ and $\pm 1$, $P_i$ must be singleton, such that the edge $e$ is not on the contour $C_i$. Hence, $e$ bounds two faces whose remaining edges either are on $C_{i-1}$ or incident to the singleton $P_i$. Hence, none is drawn vertical by the contour condition.    
\end{proof}

\begin{property2rep}[\fabristar]\hspace{-0.5em}\label{prp:collinear} 
Let $u$, $v$ and $w$ be three consecutive vertices encountered in this order in a counterclockwise walk along the boundary of an inner face of $\Gamma$. If they are collinear and the~line through them has zero slope, then they are part of a chain. If they are collinear and the~line through them has positive (negative) slope, then $y_u < y_v < y_w$ ($y_u > y_v > y_w$).  If they are not collinear, then $v \notin \Delta_{uw}$.
\end{property2rep}
\begin{proof}
Consider three vertices $u$, $v$ and $w$ that appear consecutively along the boundary of $g$ (in a counterclockwise walk of $g$) and refer to Property~\ref{prp:face} for the definition of boundary parts.
It follows that if $u$, $v$ and $w$ are collinear and  horizontally aligned, then they belong to boundary part~(\ref{f:5}), i.e., they are part of a chain. 
Further, if $u$, $v$ and $w$ are collinear and not horizontally aligned, then they belong to boundary part~(\ref{f:1}) or to~(\ref{f:3}), which implies the claimed property. On the other hand, consider now the case where $u$, $v$ and $w$ are not collinear. It follows that if $v$ belongs to boundary parts~(\ref{f:1})--(\ref{f:3}), then by the fact that $g$ is convex, we can conclude that $v \notin \Delta_{uw}$. Hence, $v$ belongs to one of boundary parts~(\ref{f:4})--(\ref{f:6}), and since each of boundary parts~(\ref{f:4}) and (\ref{f:6}) consists of a single edge, we conclude that $v$ is an extreme vertex of boundary part~(\ref{f:5}). It this case, however, $v$ does not lie inside $\Delta_{uw}$ (but, actually, on it boundary). 
\end{proof}

%======================================================
\section{Algorithm Description}\label{sec:main}
%======================================================

We now describe our approach to compute a strictly-convex drawing of a $3$-connected plane graph, assuming that its outer face has at most $5$ vertices (such a face always exists). We start with \Cref{ssec:lifting}, in which we describe the properties of what we call lifting functions and liftable drawings. A key property is that applying a non-affine transformation to a liftable drawing by means of a lifting function preserves planarity. As this tool might be of independent interest, we state it as general as possible. In \Cref{ssec:application}, we prove that drawings computed by Kant's algorithm are indeed liftable and that a transformation via a liftable function makes them internally strictly-convex except for possible horizontally-aligned paths. Up to this point, it was not needed to choose a particular lifting function; in \Cref{ssec:all} we unveil our choice. We also design a second transformation targeted to faces containing paths of horizontally-aligned vertices. The output of this step is an  internally strictly-convex drawing. The last step of the algorithm is described in \Cref{ssec:outerface}, namely a simple preprocessing in which the outer face of the input graph, which by our assumption has at most $5$ vertices, is suitably augmented with dummy vertices, whose removal from the computed drawing guarantees that all faces (including the outer one) are strictly-convex.

%======================================================
\subsection{Lifting functions}\label{ssec:lifting}
%======================================================

Given a drawing $\Gamma$ of a graph $G$ and a function $f: \mathbb{R} \mapsto \mathbb{R}$, we refer to the~drawing $\Gamma_f$ obtained by applying the transformation $(x_u,y_u) \mapsto (x_u, f(y_u))$~to~$\Gamma$ as the \emph{transformed drawing} of $\Gamma$ with respect to $f$. In~view~of~\cref{thm:lifting-lemma} below, we focus on lifting functions and liftable drawings (see \cref{def:lifting,def:liftable}).

\begin{definition}\label{def:lifting}
A function $f : \mathbb{R} \mapsto \mathbb{R}$ is \emph{lifting} if and only if
\begin{inparaenum}[(i)]
    \item\label{lif:1}$f$ is strictly-convex and strictly-increasing;
    \item\label{lit:2}$f(r) \geq r, \forall r \in \mathbb{R}$;
    \item\label{lit:3}$r \in \mathbb{N} \Rightarrow f(r) \in \mathbb{N}$.
\end{inparaenum}
\end{definition}

\noindent The next property follows directly from \cref{def:lifting}.

\begin{observation}\label{obs:horizontal-vertical}
Let $\Gamma$ be a drawing of a graph $G$. 
Given a lifting function $f$, three vertices are horizontally (vertically) aligned in $\Gamma$ if and only if they are horizontally (vertically) aligned in the transformed drawing $\Gamma_f$.
\end{observation}

\begin{definition}\label{def:liftable}
A planar straight-line grid drawing of a graph is called \emph{liftable} if for every edge $(u,v)$ there is no vertex of $G$ in $\Delta_{uv}$.
\end{definition}

\begin{theorem}\label{thm:lifting-lemma}
Let $\Gamma$ be a planar straight-line grid drawing of a plane graph $G$.
If $\Gamma$ is liftable, then its transformed drawing $\Gamma_f$ with respect to a lifting function $f$ is a planar straight-line grid  drawing of $G$ with the same planar embedding as~$\Gamma$.
\end{theorem}
\begin{proof}
Condition (\ref{lit:3}) of \Cref{def:lifting} trivially implies that $\Gamma_f$ is a grid drawing. 
We next prove that $\Gamma$ and $\Gamma_f$ have the same planar embedding. (Note that, if $\Gamma$ is not liftable, $\Gamma_f$ is not necessarily planar.) 
Since, by Condition (\ref{lit:2}) of \Cref{def:lifting}, $\Gamma_f$ is obtained from $\Gamma$ by shifting vertices upwards, the existence~in $\Gamma_f$ of an edge crossing or of a vertex having a circular order of its incident edges different than the one in $\Gamma$, implies that there exist a vertex $w$ and an edge $(u,v)$ in $G$, such that $w$ is below (above)  $(u,v)$ in $\Gamma$ and above (below) $(u,v)$ in $\Gamma_f$. 

We next argue that the situation described above is not possible. Consider a vertex $w$ and an edge $(u,v)$ of $G$ and let $s$ and $s'$ be the line-segments representing $(u,v)$ in $\Gamma$ and $\Gamma_f$. Clearly, it suffices to consider the case in which $x_u \le x_w \le x_v$. Let $p$ and $p'$ be the vertical projection of $w$ on $s$ and $s'$, respectively. Also, let $y_p$ and $y_{p'}$ be the $y$-coordinates of $p$ in $\Gamma$ and of $p'$ in $\Gamma_f$, respectively. 
Suppose $y_u \leq y_v$; the case in which $y_u > y_v$ is symmetric.

Firstly, consider the case in which $w$ is above $s$, i.e., $y_w > y_{p}$. Since $\Gamma$ is liftable, vertex $w$ does not belong to $\Delta_{uv}$. Hence,  $y_u \le y_{p} \le y_v \le y_w$. Since $f$ is strictly-increasing, it follows $f(y_u) \le f(y_{p}) \le f(y_v) \le f(y_w)$. However, $f(y_w) > y_{p'}$ implies that $w$ is above $s'$, as desired. 
Secondly, consider the case in which $w$ is below $s$ in $\Gamma$. 
Here, we distinguish three cases: $x_w=x_u$, $x_w=x_v$ and $x_u < x_w < x_v$. In the first case, we have $y_w < y_{p}=y_u$ and $y_{p'}=f(y_u)$. Since $f$ is strictly-increasing, it holds $f(y_w) < f(y_u)=y_{p'}$, i.e., $w$ is below $s'$, as desired. The second case is analogous. 
For the third case, we know that $y_w < y_p$. Since $p$ lies on $s$, by \Cref{le:basic}, $f(y_p)<y_{p'}$ holds. Also, since $f$ is strictly-increasing, it follows $f(y_w) < f(y_p)$. Thus, $f(y_w)<y_{p'}$ holds, i.e., $w$ is below $s'$, as desired.
\end{proof}

%======================================================
\subsection{Application to Kant's drawings}\label{ssec:application}
%======================================================

We now show that applying a lifting function to a drawing computed by~Kant's algorithm (see \Cref{sec:kant}) yields a drawing with several important properties. 

\begin{lemma2}\label{le:kant}
Let $\Gamma$ be a drawing of a $3$-connected plane graph $G$ computed by Kant's algorithm. Drawing $\Gamma$ is liftable.
\end{lemma2}
\begin{proof}
Consider an edge $(u,v)$ of $G$ and w.l.o.g.\ assume $y_u < y_v$ in $\Gamma$. We prove that there is no vertex $w$ in $\Delta_{uv}$. This is obvious when $x_u = x_v$, 
since $\Delta_{uv} = \emptyset$. Hence, either $x_u < x_v$ or $x_u > x_v$. Consider the former case; the latter can be treated symmetrically. Suppose for a contradiction that there exists at least one vertex (other than $u$ and $v$) in $\Delta_{uv}$. Let $w$ be the rightmost vertex out of those in $\Delta_{uv}$.
Since $w$ is in $\Delta_{uv}$ we know $y_u < y_w < y_v$. Since $y_u \neq y_v$, it follows that $(u,v)$ is not the edge $(v_1,v_2)$ of $P_0$, which, in turn, implies that vertex $w$ belongs to a path $P_i$ with $i>0$, since $y_u<y_w$. Hence, by \cref{prp:b}, $w$ has at least two incident edges $(a,w)$ and $(b,w)$, such that $y_a \le y_w$, $y_b \le y_w$, and $x_a < x_w < x_b$.  Since $b$ is to the right of $w$, the way we selected $w$ implies that~$b$ does not belong to $\Delta_{uv}$; consequently, $(w,b)$ crosses $(u,v)$, contradicting the planarity of $\Gamma$. 
\end{proof}

\noindent By combining \Cref{thm:lifting-lemma,le:kant}, we conclude the following.

\begin{theorem2rep}[\fabristar]\label{thm:hor-collinearities}
Given a $3$-connected plane graph $G$ and a lifting function $f$,~let $\Gamma_f$ be the transformed drawing of a drawing $\Gamma$ of $G$ computed by Kant's algorithm. 
Then, $\Gamma_f$ is internally-convex and planar with the same embedding as $\Gamma$. Also, if two consecutive edges of an inner face of $\Gamma_f$ form an angle $\pi$ inside this face, then these edges are horizontal.
\end{theorem2rep}
\begin{proof}
By~\cref{le:kant}, $\Gamma$ is liftable. Consequently,  by~\cref{thm:lifting-lemma}, $\Gamma_f$ has the same planar embedding as $\Gamma$, thus any  face  in $\Gamma$ exists also in $\Gamma_f$ (even though with a different drawing). Consider a counterclockwise walk along the boundary of an inner face $g$ in $\Gamma$ and let $u$, $v$ and $w$ be three consecutive vertices along this walk. Let $\alpha$ and $\alpha'$ be the angle at $v$ formed by the line-segments representing edges $(u,v)$ and $(v,w)$ inside $g$ in $\Gamma$ and in $\Gamma_f$, respectively. Since $\Gamma$ is convex,  $\alpha \le \pi$. We argue that $\alpha' \le \pi$ as well, and that if $\alpha'=\pi$, then $(u,v)$ and $(v,w)$ are horizontally aligned. We distinguish whether $u$, $v$ and $w$ are collinear or not. 

\myparagraph{Case A:} \textit{Vertices $u$, $v$ and $w$ are collinear in $\Gamma$.} By \Cref{prp:vertical-edges}, vertices $u$, $v$ and $w$ are not vertically aligned in $\Gamma$. If they are horizontally aligned in $\Gamma$, then by \Cref{obs:horizontal-vertical} they are horizontally aligned also in $\Gamma_f$, thus $\alpha=\alpha'=\pi$, as desired.
Suppose now the line $\ell$ through $u$, $v$ and $w$ in $\Gamma$ is either of positive or of negative slope, which by \Cref{prp:collinear} implies that either $y_u < y_v < y_w$ or $y_u > y_v > y_w$ holds, respectively. Then, by \Cref{le:basic}, it follows that $v$ is below the line-segment connecting $u$ and $w$ in $\Gamma_f$, hence $\alpha' < \pi$, as desired. 
    
\myparagraph{Case B:} \textit{Vertices $u$, $v$ and $w$ are not collinear in $\Gamma$.} If any two of $u$, $v$ and $w$ are horizontally aligned in $\Gamma$, since $f$ is strictly-increasing, it readily follows that $\alpha' < \pi$. So assume that no horizontal alignment exists. By symmetry we can further assume that $x_u < x_w$ holds in $\Gamma$.  We distinguish subcases based on the $x$-coordinate of $v$ in~$\Gamma$. 
\begin{itemize}
    \item $x_v < x_u$. In this case vertex $v$ is to the left of both $u$ and $w$ (in both $\Gamma$ and $\Gamma_f$), thus $\alpha' < \pi$. 

    \item $x_v > x_w$. Similarly, vertex $v$ is to the right of both $u$ and $w$ (in both $\Gamma$ and $\Gamma_f$), thus $\alpha' < \pi$. 

    \item $x_u \le x_v  \le x_w$. If $v$ is above (below) $u$ and $w$ in $\Gamma$, since $f$ is strictly-increasing, $v$ is above (below) $u$ and $w$ also in $\Gamma_f$, hence $\alpha' < \pi$. Since $u,v,w$ are not collinear, by \cref{prp:collinear}, it follows that $v \notin \Delta_{uw}$. Hence, $v$ is below the line-segment connecting $u$ and $w$ in $\Gamma$. Let $s$ and $s'$ be the line-segments connecting $u$ and $w$ in $\Gamma$ and in $\Gamma_f$, respectively. Let $p$ and $p'$ be the vertical projection of of $v$ on $s$ and $s'$, respectively. Also, let $y_p$ and $y_{p'}$ be the $y$-coordinates of $p$ in $\Gamma$ and of $p'$ in $\Gamma_f$, respectively. We know that $y_v < y_p$. Since $p$ lies on $s$, by \Cref{le:basic}, $f(y_p) < y_{p'}$ holds. Also, since $f$ is strictly-increasing, it follows $f(y_v) < f(y_p)$, hence $f(y_v)<y_{p'}$ and $\alpha' < \pi$.\qedhere
\end{itemize}
\end{proof}
\begin{proof}[Proof sketch]
Since, by \cref{le:kant}, $\Gamma$ is liftable, by \cref{thm:lifting-lemma} $\Gamma_f$ has the same planar embedding as $\Gamma$. Consider a counterclockwise walk along the boundary of an inner face $g$ in $\Gamma$ and let $u$, $v$ and $w$ be three consecutive vertices along this walk. Let $\alpha$ and $\alpha'$ be the angle at $v$ formed by the edges $(u,v)$ and $(v,w)$ inside $g$ in $\Gamma$ and in $\Gamma_f$, respectively. Since $\Gamma$ is convex,  $\alpha \le \pi$. We claim that $\alpha' \le \pi$ and that if $\alpha'=\pi$, then $(u,v)$ and $(v,w)$ are horizontally aligned in~$\Gamma_f$. We prove the claim when  $u$, $v$ and $w$ are collinear in $\Gamma$ (the other case is deferred for the appendix).
By \Cref{prp:vertical-edges}, vertices $u$, $v$ and $w$ are not vertically aligned in $\Gamma$. If they are horizontally aligned in $\Gamma$, then by \Cref{obs:horizontal-vertical} they are horizontally aligned also in $\Gamma_f$, as desired.
Suppose now the line $\ell$ through $u$, $v$ and $w$ in $\Gamma$ is either of positive or of negative slope, which by \Cref{prp:collinear} implies that either $y_u < y_v < y_w$ or $y_u > y_v > y_w$ holds, respectively. Then, by \Cref{le:basic}, it follows that $v$ is below the line-segment connecting $u$ and $w$ in $\Gamma_f$, hence $\alpha' < \pi$. 
\end{proof}

%======================================================
\subsection{Putting everything together}
\label{ssec:all}
%======================================================

We are now ready to put all pieces together. Let $G$ be a $3$-connected plane graph with $n$ vertices. Without loss of generality, we can assume that the outer face of $G$ contains at most $5$ vertices (which will be useful in the next subsection), since such a face always exists. Let $\Gamma$ be a convex drawing of  $G$ computed by Kant's algorithm and let $f:\mathbb{R} \rightarrow \mathbb{R}$ be the function  $f(y) =  5(n-2)^2 y + y^2$. Clearly, $f$ is a lifting function. Hence, by \cref{thm:hor-collinearities},  in the transformed drawing $\Gamma_f$, each inner face $g$  that is not strictly-convex contains at least three horizontally-aligned vertices. By \cref{prp:collinear}, these vertices are part of a chain in the canonical order $\delta$. Hence, by Condition~(\ref{co:4}) of canonical order, each of the vertices of this chain has at least one neighbor placed above it. By definition of $f$ it follows that each of these neighbors is positioned at least $5(n-2)^2$ units above the chain in $\Gamma_f$. We exploit this property to turn $\Gamma_f$ into an internally strictly-convex drawing by shifting all vertices of each chain upwards while keeping $\Gamma_f$ planar.

To this end, let $P_k=\{z_1,\ldots,z_p\}$ with $p \geq 2$ be a chain in the canonical order $\delta$ used to construct $\Gamma$. For $i=1,\ldots,p$, we shift vertex $z_i$ of $P_k$ by $(x_{z_i}-x_{z_1})(x_{z_p}-x_{z_i})$ units upwards; see~\cref{fig:chains-up}. It follows that if $\lambda$ is the total width of $P_k$ in $\Gamma$ (and thus also in $\Gamma_f$), then each vertex in $P_k$ is shifted by at most $\lambda^2/4$ units of length upwards, which is in turn at most $(n-2)^2$, since the total width of $\Gamma$ is at most $2n-4$, and therefore $\lambda \leq 2n-4$. Also, note that only the \emph{internal vertices} of $P_k$ are shifted (if any), i.e., only the vertices $z_i$ with $2 \le i \le p-1$. Let $\widehat{\Gamma_f}$ be the drawing obtained from $\Gamma_f$ by applying the aforementioned procedure to each chain, which we call the \emph{curved drawing} of $\Gamma_f$.
Clearly,  $\widehat{\Gamma_f}$ is a grid drawing, we can prove   that it is planar and internally strictly-convex. 

\begin{figure}[t]
	\centering
	\begin{subfigure}[b]{.32\textwidth}
	\flushleft
	\includegraphics[page=5,width=\textwidth]{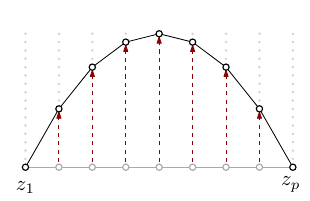}
	\subcaption{}
	\label{fig:chains-up}
	\end{subfigure}
	\hfil
	\begin{subfigure}[b]{.32\textwidth}
	\flushleft
	\includegraphics[page=2,width=\textwidth]{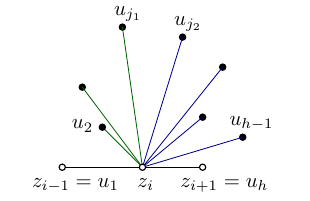}
	\subcaption{}
	\label{fig:chain-a}
	\end{subfigure}
	\hfil
	\begin{subfigure}[b]{.32\textwidth}
	\flushright
	\includegraphics[page=3,width=\textwidth]{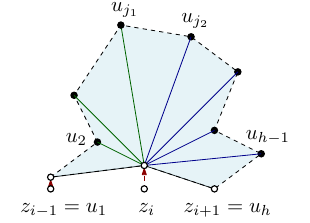}
	\subcaption{}
	\label{fig:chain-b}
	\end{subfigure}
	\caption{(a)~Illustration of the procedure of shifting the vertices of a chain upwards, and (b-c)~cases that arise in the proof of \Cref{thm:internally-convex}.}
	\label{fig:chain}
\end{figure}

\begin{toappendix}
\medskip The next lemma concludes the proof of \Cref{thm:internally-convex}.

\begin{lemma2}\label{lem:still-planar}
Given a $3$-connected $n$-vertex plane graph $G$ and the lifting function $f:\mathbb{R} \mapsto \mathbb{R}$ with $f(y) =  5(n-2)^2 y + y^2$,~let $\Gamma_f$ be the transformed drawing of a drawing $\Gamma$ of~$G$ computed by Kant's algorithm. Then, the curved drawing $\widehat{\Gamma_f}$ of $\Gamma_f$ is planar with the same embedding as $\Gamma$.
\end{lemma2}
\begin{proof}
Consider the canonical order $\delta$ of $G$ used to compute $\Gamma$. 
Let $P_k=\{z_1,\ldots,z_p\}$ be a chain of $\delta$ with $p \ge 3$, i.e., with at least one internal vertex. Consider a vertex $z_i$, with $2 \le i \le p-1$ (recall that $z_1$ and $z_p$ are not shifted when going from $\Gamma_f$ to $\widehat{\Gamma_f}$). Let $u_1 = z_{i-1}, u_2,  \dots, u_{h-1}, z_{i+1}=u_h$ be the $h$ neighbors of $z_{i}$, for some $h \ge 3$, ordered such that $(z_i,u_j)$ precedes $(z_i,u_{j+1})$, for each $1 \le j < h$, in clockwise order around $z_i$ starting at $(u_1,z_i)$. Refer to \Cref{fig:chain-a} for an illustration. We make two observations:
\begin{inparaenum}[(i)]
\item Each $u_j$, with $2 \le j \le h-1$, belongs to a distinct path of $\delta$. Let $l(u_j)$ be an index such that $u_j \in P_{l(u_j)}$. 
\item By the edge coloring described in \Cref{sec:kant}, we know that there exist two (possibly coinciding) indexes $2 \le j_1 \le j_2 \le h-1$, such that $k=l(u_1)<l(u_2)< \dots < l(u_{j_1})$ and $k=l(u_h)<l(u_{h-1})< \dots < l(u_{j_2})$. 
\end{inparaenum}
To see this, observe that each edge $(z_i,u_j)$ with $j < j_1$ is green, each edge $(z_i,u_j)$ with $j > j_2$ is blue, and either $(z_i,u_{j_1})$ is green and  $(z_i,u_{j_2})$ is blue if $j_1 < j_2$, or $(z_i,u_{j_1})$ is red if $j_1 = j_2$. Consequently, $u_1,\dots,u_{j_1-1}$ are all to the left of $z_i$ and hence we call them the \emph{left neighbors} of $z_i$, while  $u_{j_2+1},\dots,u_h$ are all to the right of $z_i$ and hence we call them the \emph{right neighbors} of $z_i$. Also, $u_{j_1}$ and $u_{j_2}$ are to the left and to the right of $z_i$, respectively, if $j_1 < j_2$, while $u_{j_1}$ can be in any position with respect to $z_i$ (including vertically aligned) if $j_1=j_2$. In any case, since $k=l(u_1)<l(u_2)< \dots < l(u_{j_1})$ and $k=l(u_h)<l(u_{h-1})< \dots < l(u_{j_2})$, we know that, in $\Gamma_f$, $y_{u_1} < y_{u_2} < \dots < y_{u_{j_1}}$ and $y_{u_h}<y_{u_{h-1}}< \dots <u_{j_2}$. Also, none of the vertices $u_i$ with $2 \le i \le h-1$ is an internal vertex of a chain (since it is adjacent to $z_i$ which belongs to a previous path in $\delta$), and hence it has not been shifted when going from $\Gamma_f$ to $\widehat{\Gamma_f}$. As a consequence, when going from $\Gamma_f$ to $\widehat{\Gamma_f}$, the edges whose drawing changes are either edges of a chain or edges having exactly one end-vertex that is part of a chain. Therefore, it suffices to prove that updating the position of the vertices of a single chain $P_k$ preserves planarity and does not change the planar embedding, because the updates of distinct chains are independent. 

In light of the above discussion, we prove that updating the position of the vertices of a chain $P_k=\{z_1,\ldots,z_p\}$ does not introduce edge crossings and actually does not change the planar embedding. If $p=2$, no vertex of $P_k$ is moved and hence the claim trivially holds. Hence assume $p>2$. We perform the update one vertex by one from $z_1$ to $z_{p-1}$ (vertex $z_p$ is not moved). The proof is by induction on the number $i$ of updated vertices. In the base case $i=1$ and there is nothing to prove since $z_1$ is not moved. Suppose now $i>0$. Before moving $z_i$ the current drawing is planar and maintains the original planar embedding by induction. We adopt the same notation as in the previous paragraph. Observe that $z_{i-1}$ has been shifted up by at most $(n-2)^2$ units while vertices $u_2$ and $u_{h-1}$ are at least $5(n-2)^2$ units above $z_i$. Consequently, in the current drawing it holds $y_{u_1} < y_{u_2} < \dots < y_{u_{j_1}}$ and $y_{u_h}<y_{u_{h-1}}< \dots <u_{j_2}$. Also, the polygon bounded by vertices $z_i$ and $u_j$ with $1 \le j \le h$ is empty; such a polygon is shaded in \Cref{fig:chain-b}. Indeed the presence of a vertex not adjacent to $z_j$ inside this polygon would contradict \Cref{prp:b}. Now, since we are only shifting up $z_i$, in order to create a crossing between an edge $(z_i,u_j)$ and another edge or to change the circular order of the edges at $u_j$, for some $1 \le j \le h$, it must be that $z_i$ is moved above $u_2$ or $u_{h-1}$. This is not possible since we have already observed that  $u_2$ and $u_{h-1}$ are at least $5(n-2)^2$ units above $z_i$.
\end{proof}
\end{toappendix}

\begin{theorem}[\fabristar]\label{thm:internally-convex}%
Given a $3$-connected $n$-vertex plane graph $G$ and the lifting function $f:\mathbb{R} \mapsto \mathbb{R}$ with $f(y) =  5(n-2)^2 y + y^2$,~let $\Gamma_f$ be the transformed drawing of a drawing $\Gamma$ of~$G$ computed by Kant's algorithm. Then, the curved drawing~$\widehat{\Gamma_f}$ of $\Gamma_f$ is an internally strictly-convex grid drawing of $G$ with the same planar embedding as $\Gamma$.
\end{theorem}
\begin{proof}[Proof sketch]
Concerning the planarity of $\widehat{\Gamma_f}$, consider any internal vertex of a chain and its neighbors. At high-level, we have that the envelope through such neighbors is a polygon whose boundary is formed by a left and a right path that are $y$-monotone (see \Cref{fig:chain-a}). Since the vertex remains below its lowest successor in the canonical order (see \Cref{fig:chain-b}), no edge crossing can be introduced. Thus, $\widehat{\Gamma_f}$ is planar with the same embedding as $\Gamma_f$ (and thus as $\Gamma$, by \cref{thm:hor-collinearities}). 

We next prove that $\widehat{\Gamma_f}$ is internally strictly-convex. Let $g$ be an inner face in $\widehat{\Gamma_f}$. Recall that only internal chain-vertices are shifted in the transition from $\Gamma_f$ to $\widehat{\Gamma_f}$. Hence, if $g$ does not contain internal chain-vertices, then it is strictly-convex by \cref{thm:hor-collinearities}. Thus, we may assume that $g$ contains at least one such vertex $z$. By \cref{prp:face}, $z$ is either in the topmost chain of black edges of $g$, or it is one of the (at most two) bottommost vertices of $g$. Consider the former, as the latter is similar. Let $P_k=\{z_1,\ldots,z_p\}$ be the chain containing $z$. We argue that any angle inside $g$ incident to these vertices is smaller than $\pi$. By construction, this is the case for vertices $z_2,\ldots,z_{p-1}$. Hence, it remains to consider the angles at $z_1$ and $z_p$. Since the two cases are symmetric, consider the angle at $z_1$. Let $w_\ell$ be the neighbor of $z_1$ along $C_{k-1}$, i.e., $w_\ell$ is the vertex preceding $z_1$ in a clockwise walk of $g$ starting from $z_1$. We will prove that the slope of $(w_\ell,z_1)$ is strictly greater than the one of $(z_1,z_2)$, hence the angle at $z_1$ is less than $\pi$. 

\begin{figure}[t]
    \centering
    \includegraphics{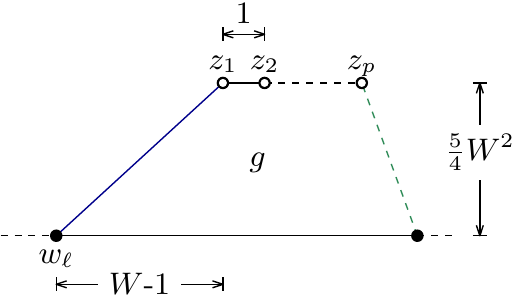}
    \caption{Illustration for the proof of \cref{thm:internally-convex}; $W$ denotes the width of $\Gamma_f$.}
    \label{fig:chain-lifting}
\end{figure}

Refer to \cref{fig:chain-lifting}. By the way the vertices of $P_k$ are shifted in the transition from $\Gamma_f$ to $\widehat{\Gamma_f}$, it follows that the maximum of the slope of $(z_1,z_2)$ is $2(n-2)-1$ (i.e., achieved when $P_k$ is of maximum $x$-length in $\widehat{\Gamma_f}$ and the $x$-distance of $z_1$ and $z_2$ in $\widehat{\Gamma_f}$ is $1$). We next argue for the slope of the edge $(w_\ell,z_1)$. Recall that vertex $z_1$ is not an interior vertex of $P_k$, which implies that it has not been shifted in the transition from $\Gamma_f$ to $\widehat{\Gamma_f}$. The same, however, does not necessarily hold for $w_\ell$. As a matter of fact, this vertex may be part of a chain, i.e., when $g$ does not contain boundary part~(\ref{f:1}) but boundary part~(\ref{f:2}) of \cref{prp:face}. This implies that it may have been shifted upwards by at most $(n-2)^2$ units in the transition from $\Gamma_f$ to $\widehat{\Gamma_f}$. The minimum of the slope of $(w_\ell,z_1)$ in $\Gamma_f$ is achieved, when $(w_\ell,z_1)$ is of maximum $x$-length in $\Gamma_f$ and of minimum $y$-length. Since the former is at most $2(n-2)-1$, while the latter is at least $5(n-2)^2$, it follows that the minimum of the slope of $(w_\ell,z_1)$ is potentially $\frac{5(n-2)^2}{2(n-2)-1}$ in $\Gamma_f$. Since in the transition from $\Gamma_f$ to $\widehat{\Gamma_f}$ vertex $w_\ell$ may be shifted by at most $(n-2)^2$ units, it follows that the slope of $(w_\ell,z_1)$ may reduce further to $\frac{5(n-2)^2-(n-2)^2}{2(n-2)-1}$, which is its minimum value. Therefore, the slope of the edge $(w_\ell,z_1)$ is strictly greater than the one of $(z_1,z_2)$,  since the following trivially holds: 
\[
\frac{5(n-2)^2-(n-2)^2}{2(n-2)-1} > 2(n-2)-1 \iff 2(n-2) > 2(n-2)-1.\qedhere
\]
\end{proof}

%======================================================
\subsection{Outer Face and Final Analysis}
\label{ssec:outerface}
%======================================================

To complete the description of our algorithm, it remains to guarantee that the outer face of the computed drawings is strictly-convex. To this aim, we slightly augment the input graph $G$ and suitably choose the canonical order to give as input to Kant's algorithm. 
Consider a planar embedding of $G$ and let $v_1,v_2,\dots,v_h$ be the vertices on the outer face (see \Cref{fi:outer-a}); recall that we have assumed $h \le 5$. If $h=3$, then  the boundary of the outer face is a triangle and hence strictly-convex. So, assume $4 \le h \le 5$. To ease the presentation, we let $h=5$  (see \Cref{fi:outer-a,fi:outer-b}), as the case  $h=4$ is simpler (see \Cref{fi:outer-c,fi:outer-d}). We proceed by adding two vertices $v_1^\star$ and $v_2^\star$ in the outer face of $G$ and edges $(v_1^\star,v_2^\star)$, $(v_1^\star,v_5)$, $(v_1^\star,v_1)$, $(v_1^\star,v_2)$, $(v_2^\star,v_3)$, $(v_2^\star,v_4)$, and $(v_2^\star,v_5)$. The resulting graph $G^\star$ 
is still planar and $3$-connected. In particular, its outer face is a $3$-cycle formed by $v_1^\star,v_2^\star,v_5$. 
We compute a canonical order $\delta^\star$ of $G^\star$ with $P_0=(v_1^\star,v_2^\star)$ and $P_m=\{v_5\}$. The key observation is that the second set of $\delta^\star$ is the chain $P_1=\{v_2,v_3\}$, since it forms the inner face $f^\star$ of $G^\star$ with $(v_1^\star,v_2^\star)$ on its boundary. 

Next, we apply the algorithm supporting \Cref{thm:internally-convex} to $G^\star$ using the aforementioned canonical order $\delta^\star$ and obtain a drawing of it that is internally strictly-convex. We next prove that the removal of $v_1^\star$ and $v_2^\star$ from this drawing yields a drawing of $G$ that is strictly-convex. By \Cref{thm:internally-convex} and by our augmentation, it suffices to guarantee that the outer face of the obtained drawing is strictly-convex. Consider first the inner angle at $v_5$ of the polygon bounding the outer face; this angle is strictly less than $\pi$, because $v_5$ is the topmost vertex of the drawing (and no other vertex is horizontally aligned with it). A similar argument applies for the angles at $v_1$ and $v_4$; in particular, after the removal of $v_1^\star$ and $v_2^\star$, vertices $v_1$ and $v_4$ are the leftmost and the rightmost neighbors of $v_5$, respectively, and therefore they are the leftmost and the rightmost vertices in the drawing, respectively. Concerning $v_2$ and $v_3$, they are horizontally aligned and, after the removal of $v^\star_1$ and $v^\star_2$, they are the bottommost vertices of the drawing. Thus, their angles are also strictly less that $\pi$ completing the proof of our claim. To conclude the proof of \cref{thm:main}, it remains to discuss the area required by the drawing obtained as above and the time complexity to compute~it.

\begin{figure}[t]
\centering
\begin{subfigure}{.24\textwidth}
\flushleft
\includegraphics[page=1,width=\textwidth]{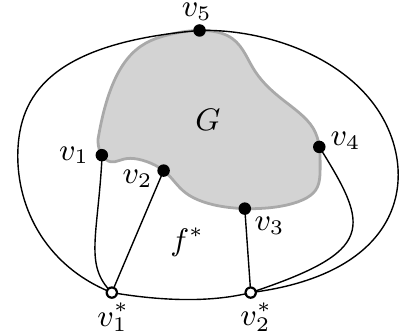}
\subcaption{}
\label{fi:outer-a}
\end{subfigure}
\begin{subfigure}{.24\textwidth}
\flushleft
\includegraphics[page=2,width=\textwidth]{figs/outer_new}
\subcaption{}
\label{fi:outer-b}
\end{subfigure}
\begin{subfigure}{.24\textwidth}
\flushleft
\includegraphics[page=3,width=\textwidth]{figs/outer_new}
\subcaption{}
\label{fi:outer-c}
\end{subfigure}
\begin{subfigure}{.24\textwidth}
\flushleft
\includegraphics[page=4,width=\textwidth]{figs/outer_new}
\subcaption{}
\label{fi:outer-d}
\end{subfigure}
\caption{Treating the outer face when its degree is five (a--b) and four (c--d).\label{fig:outerface}}
\end{figure}

\myparagraph{Area bound.} 
The drawing $\Gamma$ computed by Kant's algorithm for $G^\star$ fits on an integer grid of size $(2n^\star-4) \times (n^\star-2)$, where $n^\star=n+2$ ($G^\star$ has two more vertices than $G$). 
The transformed drawing $\Gamma_f$ of $\Gamma$ by means of the lifting function $f:\mathbb{R} \mapsto \mathbb{R}$ with $f(y) =  5(n^\star-2)^2 y + y^2$ has the same width as $\Gamma$, while the vertices $v^\star_1$ and $v^\star_2$ have $y$-coordinate $0$ in $\Gamma_f$. On the other hand, vertex $v_5$ has $y$-coordinate $ 5(n^\star-2)^2(n^\star-2)+(n^\star-2)^2$, which is also the height of $\Gamma_f$. 
Since no vertex of the outer face of $\Gamma_f$ is further shifted upwards,  the curved drawing $\widehat{\Gamma_f}$ of $\Gamma_f$ has the same width and height as $\Gamma_f$. 
After removing $v_1^\star$ and $v_2^\star$, the width of the final drawing of $G$ is at least two units less than the one of $\widehat{\Gamma_f}$, while its height is at least $5(n^\star-2)^2$ units less. Since $n^\star=n+2$, the final drawing lies on a grid of size $((2(n+2)-4)-2) \times (5((n+2)-2)^3- 4((n+2)-2)^2) = 2(n-1) \times (5n^3-4n^2)$.  

\myparagraph{Time complexity.} 
Each step of our algorithm can be implemented in $O(n)$~time: 
\begin{inparaenum}[(i)]
\item finding a planar embedding of $G$ with a face of degree at most $5$, 
\item computing $G^\star$, a canonical order of it, and applying Kant's algorithm to $G^\star$, 
\item computing the transformed drawing with respect to our lifting function $f$ and updating the position of the internal chain-vertices.
\end{inparaenum}
This completes the proof of \Cref{thm:main}.

%======================================================
\section{Conclusions and Open Problems}\label{sec:conclusions}
%======================================================
We have provided a linear-time algorithm that computes a strictly-convex drawing of a $3$-connected planar graph on an integer grid of size $2(n-1) \times (5n^3-4n^2)$. Compared to the previously best-known upper bound for such drawings~\cite{Barany2006}, we largely improve the multiplicative constants by means of an arguably simpler algorithm, which therefore has the potential to be of practical use. Along the way, we proved tools that can be of independent interest (see in particular \Cref{thm:lifting-lemma}). Some problems that stem from our research are the following:

\begin{itemize}
    \item Can we achieve a similar area bound together with a constant aspect ratio?
    \item Is $\Omega(n^4)$ a lower bound for the area requirement of strictly-convex drawings?
    \item Can we compute strictly-convex drawings in small area with good edge-vertex resolution~\cite{DBLP:journals/comgeo/BekosGMPST21,DBLP:conf/iwoca/BekosGMS22}?
\end{itemize}

\bibliographystyle{splncs03}
\bibliography{references}
\end{document}